\documentclass[a4paper]{amsart}
\usepackage{bbm}
\usepackage{graphicx}

\date{November 20, 2014}

\newtheorem{lemma}{Lemma}
\newtheorem{theorem}{Theorem}

\newtheorem{corollary}{Corollary}
\newtheorem{proposition}{Proposition}

\def\cz{\mathbb{C}} % komplexe Zahlen
\def\rz{\mathbb{R}} % reelle ZahlFen
\def\nz{\mathbb{N}} % nat"urliche Zahlen
\def\gz{\mathbb{Z}} % ganze Zahlen

\def\const{\mathrm{C}} % generische Konstante

\def\ba{\mathbf{a}}
\def\bp{\mathbf{p}}
\def\bx{\mathbf{x}}
\def\by{\mathbf{y}}
\def\bA{\mathbf{A}}
\def\bJ{\mathbf{J}} 
\def\bL{\mathbf{L}}
\def\balpha{\boldsymbol{\alpha}}
\def\bsigma{\boldsymbol{\sigma}}

\def\bomega{\boldsymbol{\omega}}
\def\bO{\mathcal{O}}

\def\cF{\mathcal{F}}
\def\cE{\mathcal{E}}
\def\cS{\mathcal{S}}

\def\gh{\mathfrak{h}}
\def\gH{\mathfrak{H}}

\def\rd{\mathrm{d}}
\def\ri{\mathrm{i}}

\def\sgn{\mathrm{sgn}}
\def\eh{\tfrac12}
\def\ct{\chi}
\def\rtf{\rho_Z^{\mathrm{TF}}}

\def\tr{\mathrm{tr}}
\def\I{\mathbbm{1}}

\title[Scott Correction]{The Ground State Energy of Heavy Atoms: the
  Leading Correction} \author{Michael Handrek}
\email{handrek@math.lmu.de} \author{Heinz Siedentop}
\email{h.s@lmu.de}
\address{Mathematisches Institut\\Ludwig-Maximilians-Universit\"at M\"unchen\\
  Theresienstr. 39\\ 80333 M\"unchen\\Germany}

\begin{document}
\maketitle
\begin{abstract}
  For heavy atoms (large atomic number $Z$) described by no-pair
  operators in the Furry picture we find the ground state's leading
  energy correction. We compare the result with (semi-)empirical
  values and Schwinger's prediction showing more than qualitative
  agreement.
\end{abstract}
\tableofcontents
\section{Introduction\label{s1}}

Since the advent of quantum mechanics the description of heavy atoms
and molecules, in particular their energies, has been of considerable
interest in physics, quantum chemistry, and mathematics.  However, as
in classical mechanics, an explicit treatment beyond one-electron
systems is elusive.  This spurred the development of
effective models of large Coulomb systems starting with Thomas
\cite{Thomas1927}, Fermi \cite{Fermi1927,Fermi1928}, and Lenz
\cite{Lenz1932} who formulated the Thomas-Fermi model in the
appropriate language of energy functionals. It asserts that in terms
of the nuclear charge $Z$ the ground state energy of an atom is of
leading order $Z^{7/3}$.  That this is only an approximation was clear
since the beginning. But it was even doubted that the leading behavior
for large $Z$ was correct. In fact, Foldy \cite{Foldy1951} claimed a
$Z^{12/5}$ behavior for large $Z$ based on numerical
computations. Scott \cite{Scott1952} offered a refined physical
argument yielding a positive, additive correction of order $Z^2$.  The
correction originates from the strongly bound electrons which are ill-represented semi-classically. Later Schwinger \cite{Schwinger1980} and
Bander \cite{Bander1982} gave additional arguments for the validity of
the Scott correction. Schwinger \cite{Schwinger1981} and Englert and
Schwinger
\cite{EnglertSchwinger1984StatisticalAtom:H,EnglertSchwinger1984StatisticalAtom:S,EnglertSchwinger1985A}
even argued for a $Z^{5/3}$ that can be partially traced back to Dirac
\cite{Dirac1930}.

The question whether the conjectured formula for the ground state
energy holds asymptotically for the $N$-electron Schr\"odinger theory
was left unanswered until Lieb and Simon's seminal paper
\cite{LiebSimon1977} which successfully established the asymptotic
correctness of Thomas', Fermi's, and Lenz's prediction.  Later the
Scott correction for atoms was mathematically confirmed by Hughes
\cite{Hughes1986,Hughes1990} (lower bound), and Siedentop and Weikard
\cite{SiedentopWeikard1987O,SiedentopWeikard1987U,SiedentopWeikard1989}
(lower and upper bound).  Eventually, even the existence of a
$Z^{5/3}$-term, the so-called Dirac-Schwinger correction, was proved
by Fefferman and Seco in a tour de force
\cite{FeffermanSeco1990O,FeffermanSeco1989,FeffermanSeco1992,FeffermanSeco1993,FeffermanSeco1994,FeffermanSeco1994T,FeffermanSeco1994Th,FeffermanSeco1995}.
These results were extended in various ways to, e.g., ions and  molecules
\cite{Bach1989,IvriiSigal1993,Balodis2004}.

However, from a physical point of view, it is questionable how these
mathematical results reflect reality.  It is expected that electrons
located close to the nucleus will move at high velocities, thus
requiring a relativistic treatment.  Already in non-relativistic
quantum mechanics the bulk of the electrons are forced on orbits of
distances on order $Z^{-1/3}$, the electrons contributing to the Scott
correction even live on the scale $Z^{-1}$. Thus, it is
expected that only the latter will generate the leading correction to the non-relativistic ground state energy.  Already
Schwinger \cite{Schwinger1980} estimated this effect and illustrated
that the leading term remains unaffected.  S{\o}rensen
\cite{Sorensen2005} was the first to put Schwinger's prediction on
mathematical ground by showing that the leading Thomas-Fermi term is
indeed left unchanged in the limit of large nuclear charge $Z$ and
large speed of light $c$ when replacing the non-relativistic
Hamiltonian by the Chandrasekhar (or pseudo-relativistic) Hamiltonian.
Subsequently, the Scott correction for the Chandrasekhar model was
again proved to be of order $Z^2$ by S{\o}rensen \cite{Sorensen1998}
(non-interacting case) and Solovej, S{\o}rensen and Spitzer
\cite{Solovejetal2008}(interacting case).  A short alternative proof
was given by Frank, Siedentop, and Warzel \cite{Franketal2008}.  By
going from the non-relativistic Schr\"odinger theory to the
pseudo-relativistic Chandrasekhar theory one observes a lowering of the
leading energy correction.

Despite the mathematical success in establishing the large
$Z$-asymptotics for this simplified relativistic model, it is still
desirable to examine models that not only represent qualitative
features of relativistic systems but are also expected to be quantitatively correct.
In particular, the pseudo-relativistic theory fails to reproduce the
energies of hydrogen-like atoms. In fact, it predicts collapse of the
innermost electron for Ra ($Z=88$) and beyond, i.e., it does not even allow to treat very large atoms.  For this matter it is
believed to be necessary to study Hamiltonians that are derived
directly from QED, among them the so called no-pair operators
\cite{Sucher1980,Sucher1984,Sucher1987}. The most simple of those
models has been introduced by Brown and Ravenhall
\cite{BrownRavenhall1951}.  In analogy to the Schr\"odinger and
Chandrasekhar theory, Cassanas and Siedentop
\cite{CassanasSiedentop2006} proved that to leading order this model
has no effect on the energy asymptotics. In a further step Frank,
Siedentop, and Warzel \cite{Franketal2009} established the Scott
correction for this model as well. However, although the
Brown-Ravenhall model raises the energies of the
Chandrasekhar model, it nevertheless -- even in the one-particle
picture -- predicts still too low energies. In consequence, the Scott correction which is determined by the pure unscreened
Coulomb potential is too small. In fact, the energy becomes unbounded
from below at about $Z=124$, which is higher than any known element,
but nevertheless lower than the expected value $137$.

In passing we would like to mention two recent results on the
inclusion of the self-generated magnetic field (quantized or not):
it turns out that it does not affect the leading order (Erd{\H{o}}s and
Solovej \cite{ErdosSolovej2010}). In fact the Scott term is not
affected in the Chandrasekhar model either whereas in the
non-relativistic setting -- which is expected to be unphysical because of
the argument mentioned above -- the Scott term would be changed
\cite{Erdosetal2012,Erdosetal2012S,Erdosetal2012Sc}. This and
corresponding numerical results motivate us, to drop the
self-generated magnetic fields in this paper.

Despite the shortcomings of the original Brown-Ravenhall model,
it is known that no-pair operators reproduce relativistic effects in a
quantitative correct manner \cite{Hess1986}, when taking into account
the external potential in determining the electron space (Furry
picture) \cite{Hess1986,SaueVisscher2003}. Physically this is
reasonable since the first order correction stemming from the
innermost electrons is now derived from the Dirac equation with the
unscreened Coulomb potential. This is actually the underlying physical
argument for Schwinger's \cite{Schwinger1980} relativistic correction.

\subsection{The Energy Form}
The relativistic description of an electron moving in the potential of
a static nucleus of charge $Z$ and a mean-field potential $\chi$ is
given by the Dirac operator
\begin{equation}
  \label{Coulomb-Dirac}
  D_{\gamma,\chi}=\balpha\cdot\bp+\beta-\frac{\gamma}{|\bx|}+\chi
\end{equation}
with $\gamma=Z/c$ in units of $mc^2$ and $\bp:=(1/\ri)\nabla$. This
operator is self-adjointly realized in $L^2(\rz^3:\cz^4)$. Here
\begin{equation}
  \balpha=\begin{pmatrix}0 & \bsigma\\ 
    \bsigma & 0\end{pmatrix},\ \ \beta=\begin{pmatrix}1 & 0 \\ 
    0 & -1\end{pmatrix}
\end{equation}
where $\bsigma=(\sigma_1,\sigma_2,\sigma_3)$ are the three Pauli
matrices in standard representation. A physical example for $\chi$
would be the screening mean-field potential of the electrons. We will
assume that $\chi$ is a bounded operator. Assuming $\gamma:=Z/c\in(-1,1)$,
Nenciu \cite{Nenciu1976} showed that $D_{\gamma,\chi}$ defined on
$\cS(\rz^3:\cz^4)$ has a distinguished self-adjoint extension whose
domain includes $H^1(\rz^3:\cz^4)$ and whose form domain is
${H^{1/2}(\rz^3:\cz^4)}$. By the usual abuse of notation we do no longer
distinguish notationally between the original operator and its
extension. Moreover, we write $D_\gamma:=D_{\gamma,0}$ for the pure
Coulomb-Dirac operator without screening.

The energy form $\cE^N$ on $\cS(\rz^{3N}:C^{4^N})$ is
\begin{equation}
\begin{aligned}
  \label{energyform}
  \cE^N[\Psi]:=&c^2\left\langle\Psi,\left(\sum_{\nu=1}^{N}(D_{\gamma}-1)_\nu+ \sum_{1\leq\mu<\nu\leq N} \frac\alpha{|\bx_\nu-\bx_\mu|}\right)\Psi\right\rangle\\
  =&\left\langle V_c\Psi,\left(\sum_{\nu=1}^{N}(c\balpha_\nu\cdot\bp_\nu+c^2\beta_\nu-\frac{Z}{|\bx_\nu|}-c^2)+ \sum_{1\leq\mu<\nu\leq N} \frac1{|\bx_\nu-\bx_\mu|}\right)V_c\Psi\right\rangle
  \end{aligned}
\end{equation}
where $\alpha:=e^2/(\hbar c)$, i.e.,
$\alpha=1/c$ in our units, is the Sommerfeld fine structure constant. The scaling map $V_c$ will be defined in section \ref{scaling}.

Since electrons are fermions, we will restrict the domain to antisymmetric
spinors $\Psi$. Moreover, in the spirit of Dirac's postulate of a
filled Dirac sea, Brown and Ravenhall \cite{BrownRavenhall1951} and
later Sucher \cite{Sucher1984,Sucher1987}, extending this idea,
formulated this mathematically by requiring that the one-electron space is
given by the positive spectral subspace of a suitable Dirac operator.
Typical choices for such operators:
\begin{description}
\item[Free picture (Brown and Ravenhall \cite{BrownRavenhall1951})]
  The free Dirac operator, i.e., $D_0$ with 
  $$\gH_0:=\Lambda_0(L^2(\rz^3:\cz^4)):= \I_{(0,\infty)}(D_0)(L^2(\rz^3:\cz^4)).$$
\item[Furry picture (Furry and Oppenheimer
  \cite{FurryOppenheimer1934})] The Dirac operator with the external
  potential, i.e., $D_\gamma$ with
  $$\gH_\gamma:=\Lambda_\gamma(L^2(\rz^3:\cz^4)):=
  \I_{(0,\infty)}(D_\gamma)(L^2(\rz^3:\cz^4)).$$
\item[Intermediate or Fuzzy picture (Mittleman \cite{Mittleman1981})]
  The intermediate picture with some screening potential $\chi$,
  e.g. it may be picked as the mean-field potential of the
  Dirac-Hartree-Fock equations. An optimal choice, which depends on
  the two-particle density matrix, was suggested by Mittleman. This
  leads to a non-linear equation, which has been studied numerically
  with great success in quantum chemistry, see, e.g., Saue
  \cite{Saue2011}.

  We will pick $\chi$ to be the rescaled Thomas-Fermi
  screening potential reduced by the exchange hole (see \eqref{loch}),
  i.e., $D_{\gamma,\chi}$, with the corresponding one-particle
  Hilbert space is
 $$\gH_{\gamma,\chi}:=\Lambda_{\gamma,\chi}(L^2(\rz^3:\cz^4)):=
 \I_{(0,\infty)}(D_{\gamma,\chi})(L^2(\rz^3:\cz^4)).$$
 \end{description}
 The corresponding $N$ particle Hilbert spaces are $\gH_{\#}^N:=
 \bigwedge_{\nu=1}^N \gH_\#$ where $\#$ either denotes the free picture,
 i.e., $\#="0"$, the Furry picture, i.e., $\#="\gamma"$, or our choice of the
 intermediate picture, i.e., $\#="\gamma,\chi"$. Because of Nenciu's result
 as described above \cite{Nenciu1976},
$$\Lambda_{\#}(\cS(\rz^3:\cz^4))\subset H^{1/2}(\rz^3:\cz^4)$$ 
and dense in $\Lambda_{\#}(L^{2}(\rz^3:\cz^4)$.Thus, the functional $\cE_{\#}^N:=
\cE^N\big|_{\bigwedge_{\nu=1}^N(\Lambda_\#(\cS(\rz^3:\cz^4))}$ is well
defined for $\gamma\in[0,1)$ and bounded screening potential, i.e., in
particular for our choice.

The above quadratic forms are bounded from below under suitable
constraints on $\gamma$, and thus define according to Friedrichs a
distinguished self-adjoint operator. In particular this holds, if
\begin{description}
\item[Brown-Ravenhall] 
  $\gamma\leq\gamma_{crit}^B=\frac{2}{2/\pi+\pi/2}$ (Evans et al
  \cite{Evansetal1996}).
\item[Furry and Intermediate] $\gamma <1$ (Nenciu \cite{Nenciu1976}).
\end{description}
Matte and Stockmeyer \cite{MatteStockmeyer2010} have worked out the detailed structure of the spectrum of such no-pair operators.
\subsection{Scaling and Units}

\label{scaling}
The technical necessity of the limit $Z\to\infty$ only being possible in the simultaneous limit $c\to\infty$ such that $\gamma=\frac{Z}{c}\to\const<1$ suggests that the appropriate units are chosen as in \eqref{Coulomb-Dirac}. 
To draw a connection to the usual Schr\"odinger operator of hydrogen
\begin{equation}
\frac{\bp^2}2-\frac{Z}{|\bx|}
\end{equation}
it will be useful to consider the unitary scaling operator $V_c$ on $L^2(\rz^{3N}:\cz^{4^N})$ defined by
\begin{equation}
 (V_c\psi)(\bx):=c^{-\frac32}\psi\left(\frac{\bx}{c}\right)
\end{equation}
such that
\begin{equation}
 c^2\left\langle \psi, \left(\balpha\cdot\bp+\beta-\frac{\gamma}{|\bx|}\right)\psi\right\rangle=\left\langle V_c\psi, \left( c\balpha\cdot\bp+c^2\beta-\frac{Z}{|\bx|}\right)V_c\psi\right\rangle.
\end{equation}

\subsection{Main Result}

Ground state energies of neutral atoms in no-pair pictures are given by
\begin{equation} E_\#(Z):=\inf \{\cE^Z[\psi]|\psi\in\mathfrak{Q}_\#^Z,\Vert\psi\Vert=1\}
\end{equation}
where
$\mathfrak{Q}_\gamma^Z:=\gH_{\gamma}^Z\cap\bigwedge_{\nu=1}^Z(\Lambda_\#(\cS(\rz^3:\cz^4))$. 

For the Brown-Ravenhall operator, i.e., $\#$ set to $0$, Cassanas and
Siedentop \cite{CassanasSiedentop2006} proved that the leading order
term is given by the Thomas-Fermi energy $\mathcal{E}^\mathrm{TF}$ (see
Appendix \ref{TFtheory}), i.e., 
\begin{equation}
E_0(Z)=E_\mathrm{TF}(1)Z^{7/3}+o(Z^{7/3}).
\end{equation}

Furthermore, Frank et al \cite{Franketal2009} found the leading
correction. It is given by the spectral shift function $s^B(\gamma)$
for $\gamma\in(0,\gamma_{crit}^B]$, i.e., the difference of bound
state energies of the one-particle Brown-Ravenhall operator and Schr\"odinger
operator.

This paper concerns the analogous result in the case of the Furry
picture, i.e., $\#$ replaced by $\gamma$. In this case we
are in the fortunate situation that the eigenvalues of the one-particle Furry operator are explicitly known. They are identical to
the eigenvalues $\lambda^D_n$ (labeled in increasing order) of the
Dirac-Coulomb operator $D_{\gamma}-1$. Likewise, denoting the
eigenvalues of the hydrogen Schr\"odinger operator
$(\bp^2/2-\gamma/|\bx|)\otimes\mathbbm{1}_{\cz^2}$ in $L^2(\rz^3:\cz^2)$ by $\lambda^S_n$, we can
write down rather explicitly a new spectral shift function
\begin{equation}
  \label{Scottfunction}
  s^D(\gamma)
  :=\frac1{\gamma^2}\sum_{n=1}^\infty \left(\lambda^D_{n}-\lambda^S_n\right)
\end{equation}
for $\gamma\in (0,1)$.

We are now ready to state the main theorem.
\begin{theorem}[Scott correction]
 \label{Scottcorrection}
 There exists a constant $C>0$ such that for all $Z>0$ and $\gamma=\frac{Z}{c}<1$ one has
 \begin{equation}
  \left|E_\gamma(Z)-E_\mathrm{TF}(Z)-\left(\frac12+s^D(\gamma)\right)Z^2\right|\leq C Z^{47/24}.
 \end{equation}

\end{theorem}
Put differently, Theorem \ref{Scottcorrection} asserts that in the limit $Z\to\infty$ one has 
\begin{equation}
 E_\gamma(Z)=E_\mathrm{TF}(Z)+\left(\frac12+s^D(\gamma)\right)Z^2+o(Z^2)
\end{equation}
uniformly in $\gamma=\frac{Z}{c}<1$.

\section{The Scott Correction of the Hydrogenic Atom: The Shift from Schr\"odinger to Dirac Energies\label{s2}}

Since the Scott correction results from the innermost electrons, it is
not only intuitively obvious that the Scott corrections can be
computed for atoms whose electrons do not interact with each other; it
is also important for obtaining the Scott correction in the
interacting case. To turn this intuition into a proof, we estimate the
difference of Dirac and Schr\"odinger eigenvalues for hydrogenic atoms
and exhibit the leading correction to the non-relativistic correction
$Z^2/2$.
\subsection{Bounds on the Energy Shift}

We denote by $\lambda^D_{\gamma,n,l,j}$ the $n^{th}$ eigenvalue of the
operator $D_{\gamma}-1$ for $0\leq\gamma< 1$ restricted to the angular
momentum subspace $\mathfrak{H}_{j,l,m}$ (note that the eigenvalue
does not depend on the azimuthal quantum number $m$). We write
$\lambda^S_{\gamma,n,l}$ for the $n^{th}$ eigenvalue of the operator
$\bp^2/2-\gamma/|\bx|$ restricted to $|Y_{l,m}\rangle\langle
Y_{l,m}|L^2(\rz^3)$. They are given by Sommerfeld's fine structure
formula \cite{Sommerfeld1916} (Gordon \cite{Gordon1928}, and Darwin
\cite{Darwin1928}) and Balmer's formula \cite{Balmer1885}
(Schr\"odinger \cite{Schrodinger1926I}) (see Bethe \cite{Bethe1933}
for a concise treatment)
\begin{align} 
  \nonumber\lambda^D_{\gamma,n,l,j}&=\left(1-\frac{\gamma^2}{\left(n+l-(j+\frac12)+\sqrt{(j+\frac12)^2-\gamma^2}\right)^2+\gamma^2}\right)^\frac12-1\\
  \label{l2} &=\left(1-\frac{\gamma^2}{\left(n+l\right)^2-2(n+l-(j+\frac12))(j+\frac12-\sqrt{(j+\frac12)^2-\gamma^2})}\right)^\frac12-1\\  
  \nonumber\lambda^S_{\gamma,n,l}&=-\frac{\gamma^2}{2(n+l)^2}
\end{align}
(Gordon \cite{Gordon1928}, 
Expanding the square root in \eqref{l2} and using 
the fact that
$\sqrt{1-x}< 1-x/2$ for $x\in(0,1]$ we have
\begin{equation}
 \lambda^D_{\gamma,n,l,j}<\lambda^S_{\gamma,n,l}<0.
\end{equation}
Moreover, for fixed $l$ and $n$ the dimension of the Dirac
eigenspace is $2(2l+1)$ (see \eqref{diml}) and the dimension of the
Schr\"odinger eigenspace is $2l+1$.
\begin{lemma}
\label{everror}
  Assume $\gamma_0<1$. Then there exists a constant $C\in\rz$ such that for all $l\in \nz$,
  $j=l\pm1/2$, $j\geq1/2$, and $\gamma\in[0,\gamma_0)$
\begin{equation}
 \label{everror2} \left|\lambda^D_{\gamma,n,l,j}-\lambda^S_{\gamma,n,l}+\frac{\gamma^4}{2\left(n+l\right)^3}\left(\frac{1}{j+\frac12}-\frac{3}{4}\frac1{n+l}\right)\right|\leq C \frac{\gamma^6}{(n+l)^4}.
\end{equation}
\end{lemma}
\begin{proof}
Note that
\begin{equation}
\frac12(n+l)^2\leq\left(n+l\right)^2-2(n+l-(j+\frac12))(j+\frac12-\sqrt{(j+\frac12)^2-\gamma^2})\leq (n+l)^2.
\end{equation}
Expanding the square root gives sufficient error estimates:
\begin{align*}
  &\left(1-\frac{\gamma^2}{\left(n+l\right)^2-2(n+l-(j+\frac12))(j+\frac12-\sqrt{(j+\frac12)^2-\gamma^2})}\right)^\frac12\\
  =&1-\frac{\gamma^2}{2\left(\left(n+l\right)^2-2(n+_l-(j+\frac12))(j+\frac12-\sqrt{(j+\frac12)^2-\gamma^2})\right)}\\
  &-\frac{\gamma^4}{8\left(\left(n+l\right)^2-2(n+l-(j+\frac12))(j+\frac12-\sqrt{(j+\frac12)^2-\gamma^2})\right)^2}+\bO\left(\frac{\gamma^6}{(n+l)^6}\right)\\
  =&1-\frac{\gamma^2}{2\left(n+l\right)^2}-\frac{\gamma^4}{2\left(n+l\right)^4}\frac{n+l-(j+\frac12)}{j+\frac12}\\
  &-\frac{\gamma^4}{8\left(\left(n+l\right)^2-2(n+l-(j+\frac12))(j+\frac12-\sqrt{(j+\frac12)^2-\gamma^2})\right)^2}\\
&+\bO\left(\frac{\gamma^6n}{(n+l)^4(l+1)^2}\right)\\
\end{align*}
\begin{align*}
  =&1-\frac{\gamma^2}{2\left(n+l\right)^2}-\frac{\gamma^4}{2\left(n+l\right)^4}\frac{n+l-(j+\frac12)}{j+\frac12}\\
  &-\frac{\gamma^4}{8\left(n+l\right)^2}+\bO\left(\frac{\gamma^6n}{(n+l)^4(l+1)}\right)\\
  =&1-\frac{\gamma^2}{2\left(n+l\right)^2}-\frac{\gamma^4}{2\left(n+l\right)^3}\left(\frac{1}{j+\frac12}-\frac{3}{4}\frac1{n+l}\right)+\bO\left(\frac{\gamma^6n}{(n+l)^4(l+1)}\right)
\end{align*}
\end{proof}
For our application it is convenient to simplify the estimate.
\begin{corollary}
\label{everror3}
Under the assumptions of Lemma \ref{everror} we have
\begin{equation}
 0\leq\lambda^S_{\gamma,n,l}-\lambda^D_{\gamma,n,l,j}\leq\frac{C\gamma^4}{\left(n+l\right)^3l}.
\end{equation}
\end{corollary}
\subsection{Expectation of the Electric Potential}

We need an estimate on the expectation value of the Coulomb potential
in eigenstates of the Coulomb-Dirac operator. In the non-relativistic case the virial theorem yields immediately $\langle\psi,\gamma/|\bx|\psi\rangle= \gamma^2(n+l)^{-2}$. In the relativistic case explicit computation (Burke and Grant \cite{BurkeGrant1967}) yields 
\begin{equation}
\label{BG}
\begin{aligned}
  &\langle\psi_{n,l,j},\frac{\gamma}{|\bx|}\psi_{n,l,j}\rangle\\
=&\gamma^2\frac{(j+\frac12)^2+(n+l-(j+\frac12))\sqrt{(j+\frac12)^2-\gamma^2}}{\sqrt{(j+\frac12)^2-\gamma^2}((\sqrt{(j+\frac12)^2-\gamma^2}+n+l-(j+\frac12))^2+\gamma^2)^{\frac{3}{2}}}
  .
\end{aligned}
\end{equation}
This formula would also follow from hypervirial theorems (see, e.g.,
\cite{Shabaev1991}).
\begin{lemma}
\label{Coulomblow}
  Pick $\gamma_0\in(0,1)$. Given $n\in\nz$, $j\in \nz-1/2$,
  $l=j\pm1/2$, $m=-j,...,j$, and $\gamma\in(0,\gamma_0]$. Let
  $\psi_{n,l,j,m}$ denote an eigenfunction with
  eigenvalue $\lambda^D_{\gamma,n,l,j}$. Then there is a constant
  $C_{\gamma_0}$ such that
\begin{equation}
  \label{fbg}
  \langle \psi_{n,l,j},\frac{\gamma}{|\bx|}\psi_{n,l,j}\rangle \leq \frac{C_{\gamma_0}\gamma^2}{(n+l)^2}.
\end{equation}
\end{lemma}
\begin{proof}
   The claim follows by noting that
\begin{align*}
  \frac1{\sqrt{(j+\frac12)^2-\gamma^2}}\leq\frac1{j+\frac12-\gamma}&\leq\frac1{(1-\gamma)(j+\frac12)},\\
  \frac1{(\sqrt{(j+\frac12)^2-\gamma^2}+n+l-(j+\frac12))^2+\gamma^2}&\leq\frac{2}{(n+l)^2},\\
  (j+\frac12)^2+(n+l-(j+\frac12))\sqrt{(j+\frac12)^2-\gamma^2}&\leq(n+l)(j+\frac12).
\end{align*}
\end{proof}
Note that the restriction $\gamma\leq \gamma_0$ is only relevant for
  $j=\frac12$.

\section{Lower Bound on the Energy}

To prove a lower bound on the energy we will reduce the multi-particle
operator to an effective one-particle operator via the pointwise
correlation inequality of Mancas et al \cite{Mancasetal2004} in the
appropriately rescaled version $x\rightarrow x/c$, i.e.,
\begin{equation}
  \label{MMS}
  \sum_{1\leq\nu<\mu\leq N}\frac1{|\bx_\nu-\bx_\mu|}\geq\sum_{\nu=1}^N \ct(\bx_\nu)-c^{-1} D[\rtf].
\end{equation}
Here $\ct$ denotes the rescaled screening part of the Thomas-Fermi potential reduced by the exchange hole, i.e.,
\begin{equation}
  \label{loch}
 \ct(\bx)=c^{-4}\int_{|\bx-\by|>R_Z(\bx/c)}\frac{\rtf(c^{-1}\by)}{|\bx-\by|}\rd \by
\end{equation}
with $R_Z$ being the unique minimal radius such that
\begin{equation}
 \int_{|\bx-\by|\leq R_Z(\bx)}\rtf(\by)\rd \by =\frac12,
\end{equation}
and
\begin{equation}
D[\rho]:=\frac12\int_{\rz^3}\int_{\rz^3}\frac{\rho(\bx)\rho(\by)}{|\bx-\by|}\rd \bx \rd \by
\end{equation}
is the classical energy of a charge density $\rho$. (The corresponding
sesquilinear form is denoted by $D(\rho,\sigma)$.)  Of course,
\begin{equation}
 0<\ct(\bx)<{1\over c^2}\int_{\rz^3}\rd \by{\rtf(\by)\over |c^{-1}\bx-\by|}
\end{equation}
and in particular
\begin{equation}
\left\Vert\ct\right\Vert_\infty\leq C Z^{\frac{4}{3}}c^{-2}
\end{equation}
for some constant $C>0$. Thus, to estimate the form $\cE^N_{\gamma}$ from
below, it suffices to bound
\begin{equation}
 \cF(d):=\tr((D_{\gamma,\chi}-1)d) - c^{-2}D[\rtf]
\end{equation}
with $d$ a one-particle density matrix on $\gH_\gamma$ of
particle number not exceeding $N$, and of finite kinetic energy, i.e.,
$0\leq d\leq1$, $\tr d\leq N$, and
$\tr(|\bp|d)<\infty$.

We now use that $\chi$ is spherically symmetric and therefore
$D_\gamma^\ct$ commutes with the total angular momentum $\bJ$ and we
split the trace into angular momentum channels, i.e.,
\begin{equation}
  \label{eq:3}
  \mathcal{F}(d)\geq \sum_{{l,j\atop l\leq L}} \tr_{j,l}((D_{\gamma}-1)d)+ \sum_{{l,j\atop l> L}} \tr_{j,l}((D_{\gamma,\chi}-1)d))-c^{-2}D[\rtf],
\end{equation}
since the projections $\Lambda_\#$ commute with the projections on the subspaces $\gh_{j,l}$ (see appendix \ref{PWA}), i.e.,
\begin{equation}
 \Pi_{j,l} \Lambda_\#=\Lambda_\# \Pi_{j,l}.
\end{equation}
The parameter $L$ will later be specified to give the desired lower bound.
\subsection{Shift from Furry to Thomas-Fermi Picture}

To prove a lower bound we use a uniform estimate on the difference
between the Furry and the Thomas-Fermi picture. 
\begin{lemma}
  \label{DifferenceFTF}
  Assume $\gamma<1$. Then for every state $\psi\in\Lambda_\gamma
  H^1(\rz^3,\cz^4)$
  \begin{equation}
    0\geq\langle\psi, \left[D_{\gamma,\chi}-1-\Lambda_{\gamma,\chi} (D_{\gamma,\chi}-1)\Lambda_{\gamma,\chi}\right]\psi\rangle\geq -{\pi^2+B(\frac14,\frac12)^2\over4\pi^2(1-\gamma^2)}\|\chi\|^2_\infty\|\psi\|^2
\end{equation}
\end{lemma}
Note that the lemma generalizes to more general positive operators $\chi$.
\begin{proof}
  The expression to be estimated can be written as a manifestly negative term
  \begin{align*}
    0\geq \langle\psi, (D_{\gamma,\chi,-} -\Lambda_{\gamma,\chi}^\perp)\psi\rangle =\langle\psi, ((\Lambda_\gamma-\Lambda_{\gamma,\chi})(D_{\gamma,\chi,-}-\Lambda_{\gamma,\chi}^\perp)(\Lambda_\gamma-\Lambda_{\gamma,\chi})  \psi\rangle.
  \end{align*}
  (In our convention the negative part $A_-$ of a self-adjoint
  operator $A$ is negative.)  By the usual integral representation
  for $\sgn(x)$, the resolvent identity, and the spectral theorem we
  get
\begin{align*}
  &\langle\psi,(\Lambda_\gamma-\Lambda_{\gamma,\chi}) |D_{\gamma,\chi,-}|(\Lambda_\gamma-\Lambda_{\gamma,\chi})\psi\rangle\\
  =&\frac1{4\pi^2}\iint_{\rz^2}\rd\mu\rd\nu\langle\psi,\frac1{D_\gamma+\ri\mu}\ct\frac1{D_{\gamma,\chi,-}+\ri\mu}|D_{\gamma,\chi,-}|\frac1{D_{\gamma,\chi,-}+\ri\nu}\ct\frac1{D_\gamma+\ri\nu}\psi\rangle\\
  \leq&\left(\frac1{2\pi}\int_\rz \rd\nu\Vert{|D_{\gamma,\chi,-}|^\frac12\over D_{\gamma,\chi,-}+\ri\nu}\ct\frac1{D_\gamma+\ri\nu}\Lambda_\gamma\psi\Vert\right)^2 \leq \left({1\over2\pi}\int_\rz\rd \nu {\|\chi\|_\infty\|\psi\|\over \sqrt{\langle\nu\rangle}  \sqrt{1-\gamma^2 +\nu^2}}\right)^2\\
  \leq &(1-\gamma^2)^{-1/2}{B(\frac 14,\frac12)^2\over
    4\pi^2}\|\chi\|_\infty^2\|\psi\|^2
\end{align*}
using the Schwarz inequality in the first estimate.

A similar estimate yields
$$ \langle\psi, (\Lambda_\gamma-\Lambda_{\gamma,\chi}) \Lambda_{\gamma,\chi}^\perp(\Lambda_\gamma-\Lambda_{\gamma,\chi})\psi\rangle\leq \tfrac14 (1-\gamma^2)^{-1}\|\chi\|_\infty^2.$$
\end{proof}
\subsection{Lower Bound in the Thomas-Fermi picture}

Lemma \ref{DifferenceFTF} and \eqref{eq:3} allow to write
\begin{align*}
  \cE^N_\gamma[\psi]/c^2  \geq &\sum_{l\leq L} \sum_{j=l\pm\frac12,j\geq0} \tr_{j,l}(\left[\Lambda_\gamma (D_\gamma-1)\Lambda_\gamma\right]_-)-\const {Z^{8/3}N\over c^4}\\
&+\sum_{L<l\leq N}\sum_{j=l\pm\frac12,j\geq0}\tr_{j,l}(\left[\Lambda_{\gamma,\chi} (D_{\gamma,\chi}-1)\Lambda_{\gamma,\chi}\right]_-)
  -D(\rtf,\rtf).
\end{align*}
By the min-max principle for operators with spectral gaps
\cite{Griesemeretal1999}, we have
\begin{equation}
  \lambda_{n,j,l,m}(\Lambda_{\gamma,\chi} (D_{\gamma,\chi}-1)\Lambda_{\gamma,\chi})\geq\lambda_{n,j,l,m}(\Lambda_0 (D_{\gamma,\chi}-1)\Lambda_0)
\end{equation}
for $l\geq1$ and $\gamma<1$,since the proof of Griesemer et al \cite{Griesemeretal1999}
extends to higher coupling constants. In particular, the condition
$$ \|(|D_0|+1)^{1/2}\Lambda_0\Lambda_{\gamma,\chi}(|D_0|+1)^{-1/2}\|_{j,l}<1$$
of \cite[Theorem 1]{Griesemeretal1999} is readily verified for $j\geq
\frac32$ following the proof of \cite[Theorem
2]{Griesemeretal1999} and using Hardy's inequality in angular momentum channel $l'$,
\begin{equation}
 \bp^2\geq \frac{(l'+\frac12)^2}{|\bx|^2},
\end{equation}
i.e.,
\begin{equation}
  \label{drehhardy}
 \bp^2\geq \frac{j^2}{|\bx|^2}
\end{equation}
in the subspaces $\gh_{j,l}$. Note that although the subspaces $\gh_{j,l}$
are not eigenspaces of the angular momentum operator $\bL$, they only
contain functions with orbital angular momentum larger than
$j-\frac12$.

Therefore,
\begin{equation}
\begin{aligned}
\label{lowerbound1}
  \cE^N_{\gamma}[\psi]/c^2\geq &\sum_{l\leq L} \sum_{j=l\pm\frac12,j\geq0} \tr_{j,l}(\left[\Lambda_\gamma (D_\gamma-1)\Lambda_\gamma\right]_-)-\const  {Z^{8/3}N\over c^4}\\
&+\sum_{L<l\leq N}\sum_{j=l\pm\frac12,j\geq0}\tr_{j,l}(\left[\Lambda_0 (D_\gamma+\chi-1)\Lambda_0\right]_-)-D[\rtf].
\end{aligned}
\end{equation}
Note that the detour via the projections $\Lambda_{\gamma,\chi}$ was central to our argument, since although we have
\begin{equation}
\lambda_{n,j,l,m}(\Lambda_{\gamma} (D_{\gamma}-1)\Lambda_{\gamma})\geq\lambda_{n,j,l,m}(\Lambda_0 (D_{\gamma}-1)\Lambda_0).
\end{equation}
by the same arguments as above, it is however not clear whether
\begin{equation}
\lambda_{n,j,l,m}(\Lambda_{\gamma} (D_{\gamma,\chi}-1)\Lambda_{\gamma})\geq\lambda_{n,j,l,m}(\Lambda_0 (D_{\gamma,\chi}-1)\Lambda_0),
\end{equation}
because of the (albeit small) perturbation $\chi$.
\subsection{Difference between Schr\"odinger and Brown-Ravenhall Energies}

We will now compare the lower bound in \eqref{lowerbound1} with the non-relativistic atomic energy, which as
the proof of the Scott correction for the Brown Ravenhall operator in
\cite{Franketal2009} shows, is given -- for $N=Z$ and $L=[Z^{1/12}]$ -- by
\begin{equation}
\begin{aligned}
\label{lowerbound2}
  E^Z_S=&-2\sum_{l=0}^{L-1} (2l+1)\sum_{n=1}^\infty{Z^2\over2(n+l)^2}-D[\rtf]\\
  &+c^2\sum_{l=L}^Z\sum_{j=l\pm\frac12,j\geq0}\tr_{j,l}(\left[\Lambda_0 (D_{\gamma,\chi}-1)\Lambda_0\right]_-)+\mathcal O(Z^\frac{5}{3}).
\end{aligned}
\end{equation}
In particular, the extension of Theorem 3.1 in \cite{Franketal2009} to all $\gamma<1$ for higher momentum channels is again easily achieved by exploiting inequality \ref{drehhardy} such that we have

\begin{theorem}\emph{\cite[Theorem 3.1]{Franketal2009}} There exists a constant $C<\infty$ such that for any $\gamma\leq1$ $v:\left[0,\infty\right)\to\left[0,\infty\right)$, satisfying
 \begin{equation}
  v(r)\leq\frac{\gamma}{r},
 \end{equation}
 any $\mu>0$ and any $l\in\nz$, $j=l\pm\frac12$ one has
 \begin{equation}
  \tr_{j,l}\left[\Lambda_0 (D_0-v(|\bx|))\Lambda_0-1+\mu\right]_--\tr_l\left[\frac{\bp^2}2-v(|\bx|)+\mu\right]_-\leq C\frac{\gamma^4}{j^2}.
 \end{equation}

\end{theorem}
Combining inequality \ref{lowerbound1} and equality \ref{lowerbound2} with the known asymptotic expansion of the non-relativistic ground state energy gives the desired lower bound for the Furry Hamiltonian.

\section{Upper Bound on the Energy}
In this section we will derive a sufficient decay of the difference of the Furry and Brown-Ravenhall operators when restricted to angular momentum $l$. For high angular momenta, up to an error of lower order the Furry picture will then be replaced by the Brown-Ravenhall picture which in turn was shown to give the correct upper bound in \cite{Franketal2009}. 
We will frequently use the following inequality.
\begin{lemma}
Assume $j\geq \frac32$ and $\gamma\leq 1$. In $\mathfrak{H}_{j,l}$
\begin{equation}
\left\Vert|D_0|^{\frac1{2}}|D_\gamma|^{-\frac1{2}}\right\Vert_{j,l}\leq \sqrt{\frac{j}{j-\gamma}}.
\end{equation}
\label{lemma2}
\end{lemma}
\begin{proof}
Note that
$
|D_0|=\sqrt{\bp^2+1}
$
and
\begin{equation}
\left\Vert|D_0|^{\frac1{2}}|D_\gamma|^{-\frac1{2}}\right\Vert_{j,l}=\sup_{\psi\in \Pi_{j,l}C_0^\infty(\rz^3,\cz^4)}\frac{\left\Vert|D_0|^{\frac1{2}}\psi\right\Vert_2}{\left\Vert|D_\gamma|^{\frac1{2}}\psi\right\Vert_2}.
\end{equation}
Squaring the operators yields
\begin{equation}   
\begin{aligned}
&\left\Vert |D_\gamma|\psi\right\Vert_{j,l} =\left\Vert (D_0-\frac\gamma{|\bx|})\psi\right\Vert_{j,l} \\ 
\geq &\left\Vert |D_0|\psi\right\Vert_{j,l}-\left\Vert\frac\gamma{|\bx|}\psi\right\Vert_{j,l}\geq (1-\frac\gamma{j}) \left\Vert |D_0|\psi\right\Vert_{j,l}
\end{aligned}
\end{equation}
The claim follows from the operator monotony of the square root.
\end{proof}
\subsection{Estimate on the Electric Potential} We will now show that for high angular momenta the expectation value of the potential in the Furry picture is close to that in the Brown-Ravenhall picture.

\begin{lemma}
\label{Coulombestimate}
  Let $j\geq\frac32$, $\psi\in\Pi_{j,l} H^{1}(\rz^3,\cz^4)$,
  $-\frac{1}{|\bx|}\leq\phi(\bx)\leq \frac{1}{|\bx|}$,
  $\gamma<1$. Then
\begin{equation}
\label{12}
|\langle\psi,(\Lambda_\gamma\phi\Lambda_\gamma-\Lambda_0\phi\Lambda_0)\psi\rangle|\leq \frac{\const}{l^2
}\langle\psi,\bp^2\psi\rangle.
\end{equation} 
\end{lemma}

\begin{proof}
  The proof follows loosely the arguments found in
  \cite{Griesemeretal1999}.

One has
\begin{equation}
\begin{aligned}
\label{14}
&|\langle\psi,(\Lambda_\gamma\phi\Lambda_\gamma-\Lambda_0\phi\Lambda_0)\psi\rangle|\\ \leq &2\underbrace{|\langle\psi,(\Lambda_\gamma-\Lambda_0)\phi\Lambda_0\psi\rangle|}_{I}+\underbrace{|\langle\psi,(\Lambda_\gamma-\Lambda_0)\phi(\Lambda_\gamma-\Lambda_0)\psi\rangle|}_{II}.
\end{aligned}
\end{equation}
The second resolvent identity gives
\begin{equation}
\label{47}
\begin{aligned}
  \Lambda_0-\Lambda_\gamma&=-\frac \gamma{2\pi}\int_{-\infty}^{\infty}\rd z\frac1{D_0-\ri z}\frac1{|\bx|}\frac1{D_\gamma-\ri z}\\
  &=-\frac{\gamma}{\pi}\int_{0}^\infty\rd
  z\frac1{D_0^2+z^2}\left(D_0\frac1{|\bx|}D_\gamma-\frac{z^2}{|\bx|}\right)\frac1{D_\gamma^2+z^2}.
\end{aligned}
\end{equation}
Term $I$:
\begin{equation}
  \label{48}
\begin{aligned}
&|\langle\psi,(\Lambda_\gamma-\Lambda_0)\phi\Lambda_0\psi\rangle|\\=&|\langle\psi,\frac\gamma{\pi}\int_{0}^\infty\rd z\frac1{D_0^2+z^2}\left(D_0\frac1{|\bx|}D_\gamma-\frac{z^2}{|\bx|}\right)\frac1{D_\gamma^2+z^2}\phi\Lambda_0\psi\rangle|.
\end{aligned}
\end{equation}
The first summand on the right side of \eqref{48} is estimated as follows
\begin{equation}
\begin{aligned}
&|\langle\psi,\frac\gamma{\pi}\int_{0}^\infty\rd z\frac1{D_0^2+z^2}D_0\frac1{|\bx|}D_\gamma\frac1{D_\gamma^2+z^2}\phi\Lambda_0\psi\rangle|\\
\leq&\frac\gamma{\pi}\left[\int_0^\infty\rd z\left\Vert\frac{1}{|\bx|}D_0\frac1{D_0^2+z^2}\psi\right\Vert ^2\right]^{\frac1{2}}\left[\int_0^\infty\rd z\left\Vert D_\gamma\frac1{D_\gamma^2+z^2}\phi\Lambda_0\psi\right\Vert ^2\right]^{\frac1{2}}.
\end{aligned}
\end{equation}
Note that
\begin{equation}
\int_0^\infty\rd z \frac1{(1+z^2)^2}=\int_0^\infty\rd z \frac{z^2}{(1+z^2)^2}=\frac{\pi}{4}.
\end{equation}
Thus,
\begin{equation}
  \label{51}
\begin{aligned}
  &\int_0^\infty\rd z\left\Vert\frac{1}{|\bx|}D_0\frac1{D_0^2+z^2}\psi\right\Vert^2\leq\frac1{j^2}\int_0^\infty\rd z\left\Vert |\bp||D_0|\frac1{D_0^2+z^2}\psi\right\Vert ^2\\
  \leq&\frac{\pi}{4j^2}\left\Vert
    |\bp||D_0|^{-\frac12}\psi\right\Vert^2
  \leq\frac{\pi}{4j^2}\left\Vert |\bp|\psi\right\Vert^2.
\end{aligned}
\end{equation}
Similarly,
\begin{equation}
  \label{52}
  \begin{aligned}
    &\int_0^\infty\rd z\left\Vert D_\gamma\frac1{D_\gamma^2+z^2}\phi\Lambda_0\psi\right\Vert ^2 = {\pi\over4} \left\Vert |D_\gamma|^{-\frac12}\phi\Lambda_0\psi\right\Vert ^2\\
    \leq&{\pi\over4}\frac{j}{j-1}\left\Vert
      |D_0|^{-\frac12}\phi\Lambda_0\psi\right\Vert
    ^2\leq\frac\pi{4j(j-1)}\left\Vert \bp\psi\right\Vert^2
  \end{aligned}
\end{equation}
where the last line uses Hardy's inequality. 

The second summand of \eqref{48}
\begin{equation}
\begin{aligned}
  &|\langle\psi,\frac \gamma{\pi}\int_{0}^\infty\rd z\frac1{D_0^2+z^2}\frac{z^2}{|\bx|}\frac1{D_\gamma^2+z^2}\phi\Lambda_0\psi\rangle|\\
  \leq&\frac \gamma{\pi}\left[\int_0^\infty\rd
    z\left\Vert\frac1{|\bx|}\frac{z}{D_0^2+z^2}\psi\right\Vert
    ^2\right]^{\frac1{2}}\left[\int_0^\infty\rd
    z\left\Vert\frac{z}{D_\gamma^2+z^2}\phi\Lambda_0\psi\right\Vert
    ^2\right]^{\frac1{2}}.
\end{aligned}
\end{equation}
The first factor on the right side is estimated similarly as in
\eqref{51}; the second factor is estimated similarly as in \eqref{52}
with the same results. Putting everything together gives
\begin{equation}
  \label{eq:sum}
  I \leq {\gamma\over2j(j-1)} \|\bp\psi\|^2
\end{equation}

Term $II$: The square root of $II$ is
\begin{equation}
 \begin{aligned}
  &|\langle\psi,(\Lambda_\gamma-\Lambda_0)\phi(\Lambda_\gamma-\Lambda_0)\psi\rangle|^{1/2}
\leq\frac1{j^{1\over2}}\left\Vert|\bp|^{1\over2}(\Lambda_\gamma-\Lambda_0)\psi\right\Vert\\
\leq&\frac1{j^{1\over2}}\sup\limits_{\left\Vert h\right\Vert=1 \atop h\in\mathfrak{H}_{j,l}}|\langle h, \int_{0}^\infty\rd z|\bp|^\frac1{2}\frac1{D_\gamma^2+z^2}\left(D_\gamma\frac1{|\bx|}D_0-\frac{z^2}{|\bx|}\right)\frac1{D_0^2+z^2}\psi\rangle|.
 \end{aligned}
\end{equation}
using \eqref{47} in the last step.  The second term of the right side yields
\begin{equation}
  \begin{aligned}
    &|\langle h, \int_{0}^\infty\rd z|\bp|^\frac1{2}\frac1{D_\gamma^2+z^2}\frac{z^2}{|\bx|}\frac1{D_0^2+z^2}\psi\rangle| \\
    \leq &\left[\int_{0}^\infty\left\Vert \frac{z}{D_\gamma^2+z^2}|\bp|^\frac1{2} h\right\Vert^2 \rd z\right]^{\frac12}\left[\int_{0}^\infty\left\Vert \frac1{|\bx|}\frac{z}{D_0^2+z^2}\psi\right\Vert^2\rd z\right]^{\frac12}\\
    \leq &\frac{\sqrt{\pi}}{2j}\left\Vert \frac{1}{|D_\gamma|^{\frac12}}|\bp|^\frac1{2} h\right\Vert\left[\int_{0}^\infty\left\Vert \bp\frac{z}{D_0^2+z^2}\psi\right\Vert^2\rd z\right]^{\frac12}\\
    \leq &\frac{\pi}{4\sqrt{j(j-1)}}\left\Vert h\right\Vert\left\Vert \bp\frac1{|D_0|^\frac12}\psi\right\Vert
    \leq \frac{\pi}{4\sqrt{j(j-1)}}\left\Vert
      \bp\psi\right\Vert.
  \end{aligned}
\end{equation}
The first term can be treated analogously.

Summing $I$ and $II$ yields the claimed estimate since $j+\frac12\geq l \geq j-\frac12$ in $\mathfrak{H}_{j,l}$.
\end{proof}
\subsection{Estimate on the Projected Dirac-Coulomb Operator}

We will need a similar estimate for the expectation value of the Dirac
operator $D_\gamma$.
\begin{lemma}
\label{kineticestimate}
Assume $\gamma\in[0,1]$, $j\geq\frac32$ and $\psi\in\Pi_{l,j} \Lambda_0H^{1}(\rz^3,\cz^4)$. Then
\begin{equation}
\label{12a}
0\leq\langle\psi,\Lambda_\gamma(D_\gamma-1)\Lambda_\gamma\psi\rangle
- \langle\psi,(D_\gamma-1)\psi\rangle \leq \frac{\const}{2l^2}\langle\psi,\bp^2\psi\rangle.
\end{equation} 
\end{lemma}
\begin{proof}
Observe that
\begin{align*}
&\langle\psi,\Lambda_\gamma(D_\gamma-1)\Lambda_\gamma\psi\rangle=\langle\psi,\Lambda_0\Lambda_\gamma(D_\gamma-1)\Lambda_\gamma\Lambda_0\psi\rangle\\
=&\langle\psi,\Lambda_0(1-\Lambda_\gamma^\perp)(D_\gamma-1)(1-\Lambda_\gamma^\perp)\Lambda_0\psi\rangle\\
=&\langle\psi,\Lambda_0(D_\gamma-1)\Lambda_0\psi\rangle-\langle\psi,\Lambda_0\Lambda_\gamma^\perp(D_\gamma-1)\Lambda_\gamma^\perp\Lambda_0\psi\rangle\\
=&\langle\psi,\Lambda_0(D_\gamma-1)\Lambda_0\psi\rangle-\langle\psi,(\Lambda_0-\Lambda_\gamma)(\left[D_\gamma\right]_--1)(\Lambda_0-\Lambda_\gamma)\psi\rangle.
\end{align*}
First of all, we note that
\begin{equation}
 \left[D_\gamma\right]_--1<0
\end{equation}
which allows us to make the crude estimates
\begin{equation}
  |\left[D_\gamma\right]_--1|= -\left[D_\gamma\right]_-+1 \leq |D_\gamma|+1\leq (1+\frac1{\sqrt{1-\gamma^2}})|D_\gamma|.
\end{equation}
Hence,

\begin{align*}
  |\langle\psi,(\Lambda_0-\Lambda_\gamma)(\left[D_\gamma\right]_--1)(\Lambda_0-\Lambda_\gamma)\psi\rangle|\leq(1+\frac1{\sqrt{1-\gamma^2}})\Vert|D_\gamma|^\frac12(\Lambda_0-\Lambda_\gamma)\psi\Vert^2.
\end{align*}
Now, 
\begin{align*}
A:=\Vert|D_\gamma|^\frac12(\Lambda_0-\Lambda_\gamma)\psi\Vert&=\frac{\gamma}{2\pi}\int_{-\infty}^{+\infty}\rd\nu\Vert|D_\gamma|^\frac12\frac1{D_\gamma+\ri\nu}\frac1{|\cdot|}\frac1{D_0+\ri\nu}\psi\Vert\\
&\leq\frac{\gamma}{2\pi}\int_{-\infty}^{+\infty}\rd\nu\Vert|D_\gamma|^\frac12\frac1{D_\gamma+\ri\nu}\Vert\Vert\frac1{|\cdot|}\frac1{D_0+\ri\nu}\psi\Vert
\end{align*}
Since $l\geq1$, we have $\langle\psi, D_\gamma\psi\rangle \geq
1/\sqrt2$ (see \eqref{l2}) and
$$ {|D_\gamma|\over D_\gamma^2 + \nu^2} \leq {\max\{1/\sqrt2,|\nu|\}\over \max\{1/2,\nu^2\} +\nu^2}.$$
Using Hardy's inequality we get
\begin{align*}
  A&\leq\frac{\gamma}{2\pi j}\int_{-\infty}^{+\infty}\rd\nu \sqrt{{\max\{\sqrt{1/2},|\nu|\}\over \max\{1/2,\nu^2\} +\nu^2}}\Vert\bp\frac1{D_0+\ri\nu}\psi\Vert\\
  &\leq\frac{\gamma}{2\pi j}\left(\int_{-\infty}^{+\infty}\rd\nu\frac{\max(\sqrt{1/2},|\nu|)}{\max(1/2,\nu^2)+\nu^2}\frac1{|\nu|^\frac12}\int_{-\infty}^{+\infty}\rd\mu|\mu|^\frac12\Vert\bp\frac1{D_0+\ri\mu}\psi\Vert^2\right)^\frac12\\
  &\leq\frac{\sqrt[8]{2}\gamma}{2\pi j}\left(\underbrace{\int_{-\infty}^{+\infty}\rd\nu\frac{\max(1,|\nu|)}{\max(1,\nu^2)+\nu^2}\frac1{|\nu|^\frac12}}_{\leq 6}\int_{-\infty}^{+\infty}\rd\mu|\mu|^\frac12\Vert\bp\frac1{D_0+\ri\mu}\psi\Vert^2\right)^\frac12\\
  &\leq
  \frac{\gamma\sqrt6}{2\pi j}\left\langle\psi,\bp\frac1{|D_0|^\frac12}\bp\psi\right\rangle^\frac12
  \leq {\gamma\sqrt6\over2\pi j}\|\bp\psi\|.
\end{align*}

\end{proof}
\subsection{Inserting the Trial Density Matrix}

To finish the proof of the upper bound we will insert the trial
density matrix constructed in Appendix \ref{TDM} (in the correct units
$V_c^*dV_c$) into the Hartree-Fock functional (without exchange energy)
\begin{equation}
  \label{eq:hf}
  \cE_\mathrm{HF}(d):= \tr\left[ \left(c\balpha\cdot\bp+c^2\beta
  -c^2-\frac{Z}{|\bx|}\right)d \right]+ D[\rho]
\end{equation}
where $\rho(\bx):=\tr_{\cz^4}d(\bx,\bx)$. We will need the following
auxiliary lemma on the Schr\"odinger energy.

\begin{lemma}
  \label{Schrenergery}
  For $\gamma:=\frac{Z}c<1$
 \begin{equation}
   E_S(Z)\geq\tr[(\frac{\bp^2}2-\frac{Z}{|\bx|})d^S_<]
   +\tr[(c\balpha\cdot\bp+c^2\beta-c^2-\frac{Z}{|\bx|})d_{>,\Phi}]+D\left[\rho_{>,\Phi}\right]+\bO(Z^\frac{47}{24})
 \end{equation}
 with $$\rho(\bx)_{>,\Phi}:=\tr_{\cz^4}(d_{>,\Phi}(\bx,\bx)).$$
\end{lemma}
\begin{proof}
  We recall, following \cite{Franketal2008} and \cite[Proposition
  4]{SiedentopWeikard1987O},
\begin{equation}
 E_S(Z)=\tr[(\frac{\bp^2}2-\frac{Z}{|\bx|})d^S]+D_(\rho^S,\rho^S)+\bO(Z^{47/24})
\end{equation}
where
$\rho^S(\bx):=\tr_{\cz^2}(d^S(\bx,\bx))=\tr_{\cz^2}(d^S_<(\bx,\bx))+\underbrace{\tr_{\cz^2}(d_>(\bx,\bx)}_{=:\rho_>(\bx))}$.
We drop the Coulomb interaction between orbitals with small angular momenta and the total density, since it is positive.  Furthermore, by
\cite[Lemma 4.7.]{Franketal2009}, it holds
\begin{equation}
 D(\rho_{>,\Phi}-\rho_>, \rho_{>,\Phi}+\rho_>)=\bO(Z^{5/3}).
\end{equation}
Note that in the reference the statement is given in terms of the transformation $U_c(\bA):=\Phi_0(\bp/c)\bA\Phi_0(\bp/c)+\Phi_0(\bp/c)\bA\Phi_0(\bp/c)$ and the density
\begin{equation}
 \rho_{U,>}(\bx):=\tr_{\cz^2}(U_c (d_{>})(\bx,\bx)).
\end{equation}
However,
\begin{equation}
 \rho_{U,>}\equiv\rho_{>,\Phi}.
\end{equation}

Finally, by Frank et al \cite[Lemma 4.6.]{Franketal2009}
\begin{equation}
\tr[\frac{Z}{|\bx|}d_{>,\Phi}]-\tr[\frac{Z}{|\bx|}d_>]=\bO(Z^{23/12})
\end{equation}
and the observation that
\begin{equation}
 \tr[(c\balpha\cdot\bp+c^2\beta -c^2)d_{>,\Phi}]=\tr[(\sqrt{c^2\bp^2+c^4}-c^2)d_>]\leq\frac12\tr[\bp^2d_>]
\end{equation}
we conclude the proof.
\end{proof}
Now, $\cE_\mathrm{HF}$ yields the following estimate.
\begin{proposition}
  \label{HartreeFock}
  For $\frac{Z}{c}<1$,
  \begin{align}
    E_{\gamma}(Z)\leq \tr[ (c\balpha\cdot\bp+c^2\beta
    -c^2-\frac{Z}{|\bx|})\Lambda_{Z,c}d ]+D(\rho^F,\rho^F)
  \end{align}
  where
  \begin{equation}
    \Lambda_{Z,c}:=\I_{(0,\infty)}(c\balpha\cdot\bp+c^2\beta)
  \end{equation}
  and
  \begin{equation}
    \rho^F(\bx):= \tr_{\cz^4}((\Lambda_{Z,c}d\Lambda_{Z,c})(\bx,\bx)).
  \end{equation}
\end{proposition}
Combining Lemma \ref{Schrenergery}, Proposition \ref{HartreeFock} and
the positivity of the Coulomb energy we find
\begin{align}
\label{energydiff}
 E_{\gamma}(Z)-E_S(Z)\leq&\tr[ (c\balpha\cdot\bp+c^2\beta -c^2-\frac{Z}{|\bx|})\Lambda_{Z,c}d_<]-\tr[(\frac{\bp^2}2-\frac{Z}{|\bx|})d^S_<]\\
 &+\underbrace{\tr[(c\balpha\cdot\bp+c^2\beta-c^2-\frac{Z}{|\bx|})(\Lambda_{Z,c}d_{>,\Phi}-d_{>,\Phi})]}_{R_1}\\
 &+D(\rho^F-\rho_{>,\Phi},\rho^F+\rho_{>,\Phi})+\bO(Z^{47/24}).
\end{align}
Moreover, we write
\begin{equation}
  D(\rho^F-\rho_{>,\Phi},\rho^F+\rho_{>,\Phi})=\underbrace{D(\rho^F_>-\rho_{>,\Phi},\rho^F_>+\rho_{>,\phi})}_{=:R_2}+\underbrace{D(\rho^F_<,\rho^F+\rho^F_>)}_{=:R_3}
\end{equation}
with
\begin{equation}
 \rho^F_<(\bx):=\tr_{\cz^4}((\Lambda_{Z,c}d_<\Lambda_{Z,c})(\bx,\bx))=\tr_{\cz^4}(d_<(\bx,\bx))
\end{equation}
and
\begin{equation}
 \rho^F_>(\bx):=\tr_{\cz^4}((\Lambda_{Z,c}d_{>,\phi}\Lambda_{Z,c})(\bx,\bx)).
\end{equation}
We will see that the first term on the right side of Inequality
\eqref{energydiff} yields the Scott correction. In the following
we will prove that the error terms $R_1, R_2$, and $R_3$ are negligible.
\begin{lemma}
  As $Z\to\infty$
  \begin{equation}
    R_1=\tr[ (c\balpha\cdot\bp+c^2\beta -c^2-\frac{Z}{|\bx|})(\Lambda_{Z,c}d_{>,\Phi}-d_{>,\Phi})]=\bO(Z^{23/12})
  \end{equation}
  uniformly in $\gamma\in [0,1)$
\end{lemma}
\begin{proof}
  By Lemmata \ref{Coulombestimate} and \ref{kineticestimate} and the
  scaling behavior of the map $V_c$ we have
  \begin{equation}
  \begin{aligned}
    0&\leq\tr[ (c\balpha\cdot\bp+c^2\beta
    -c^2-\frac{Z}{|\bx|})\Lambda_{Z,c}d_{>,\Phi}]
    -\tr[ (c\balpha\cdot\bp+c^2\beta -c^2-\frac{Z}{|\bx|})d_{>,\Phi}]\\
    &\leq
    \sum_l\frac{C}{l^2}\tr_l[\bp^2d_{>,\Phi}]=\sum_l\frac{C}{l^2}\tr_l[\bp^2d_{>}]
  \end{aligned}
  \end{equation}
  The result follows from Lemma \ref{lemmae1}.
\end{proof}
Likewise we can show that $R_2$ is negligible.
\begin{lemma}
 As $Z\to\infty$ 
 \begin{equation}
  D(\rho^F_>-\rho_{>,\Phi},\rho^F_{>}+\rho_{>,\phi})=\bO(Z^{23/12}).
 \end{equation}
 uniformly in $\gamma\in [0,1)$.
\end{lemma}
\begin{proof}
  By Newton's theorem and, in the last step, Lemma
  \ref{Coulombestimate} and scaling,
 \begin{equation}
   \begin{aligned}
     D(\rho^F_>-\rho_{>,\Phi},\rho^F_{>}+\rho_{>,\phi})&\leq\frac12\int(\rho^F_{>}(\bx)+\rho_{>,\phi}(\bx))\rd \bx\int\frac{\rho^F_{>}(\bx)-\rho_{>,\Phi}(\bx)}{|\bx|}\rd \bx\\
     &\leq Z \tr[ \frac{1}{|\bx|}(\Lambda_{Z,c}d_{>,\Phi}\Lambda_{Z,c}-d_{>,\Phi})]
     \leq \sum_l\frac{C}{l^2}\tr_l[\bp^2d_{>}].
   \end{aligned}
 \end{equation}
 The lemma follows again from Lemma \ref{lemmae1}.
\end{proof}
It remains the error estimate for $R_3$.
\begin{lemma}
 As $Z\to\infty$ 
 \begin{equation}
  D(\rho^F_<,\rho^F+\rho^F_{>})=\bO(Z^{11/6})
 \end{equation}
 uniformly in $\gamma\in [0,1)$.
\end{lemma}
\begin{proof}
 Again, by Newton's theorem,
 \begin{equation}
   D(\rho^F_<,\rho^F+\rho^F_{>})\leq\frac12\int(\rho^F(\bx)+\rho^F_{>}(\bx))\rd \bx\int\frac{\rho^F_<(\bx)}{|\bx|}\rd \bx\leq
   Z\int\frac{\rho^F_<(\bx)}{|\bx|}\rd \bx.
 \end{equation}
 The Coulomb energy of the electrons with low angular momenta can be
 estimated with help of Lemma \ref{Coulomblow} and scaling
\begin{equation}
\begin{aligned}
 &\int\frac{\rho^F_<(\bx)}{|\bx|}\rd \bx=\tr[ \frac{1}{|\bx|}d_<]=\sum_{l=0}^{L-1}\sum_{j\geq\frac12, j=l\pm\frac12} \sum_{m=-j}^j\sum_{n=1}^{K-l} \langle\psi_{j,l,m,n},\frac1{|\bx|}\psi_{j,l,m,n}\rangle\\
 \leq&\sum_{l=0}^{L-1}\sum_{j\geq\frac12, j=l\pm\frac12} \sum_{m=-j}^j\sum_{n=1}^{K-l} \frac{CZ}{(n+l)^2}\\
\leq&\frac{C}{Z}\sum_{l=0}^{L-1}(2l+1)\sum_{n=1}^{K-l}\frac1{(n+l)^2}\leq CL^2K^{-1}Z=\bO(Z^{5/6}).
 \end{aligned}
\end{equation}

\end{proof}

\section{Finishing the Proof}
In the previous section we proved
\begin{multline}
  E_\gamma(Z)-E_S(Z)\\
  \leq\tr[ (c\balpha\cdot\bp+c^2\beta
  -c^2-\frac{Z}{|\bx|})\Lambda_{Z,c}d_<]-\tr[(\frac{\bp^2}2-\frac{Z}{|\bx|})d^S_<]+\bO(Z^{47/24}).
\end{multline}
Upon scaling with the map $V_{c^{-1}}$ we have
\begin{equation}
\begin{aligned}
 &\tr[ (c\balpha\cdot\bp+c^2\beta -c^2-\frac{Z}{|\bx|})\Lambda_{Z,c}d_<]-\tr[(\frac{\bp^2}2-\frac{Z}{|\bx|})d^S_<]\\
=&{Z^2\over\gamma^2}\sum_{l=0}^{L-1}\sum_{j\geq\frac12, j=l\pm\frac12}(2j+1)\sum_{n=1}^{K-l}(\lambda^D_{\gamma,n,l,j}-\lambda^S_{\gamma,n,l})
 =Z^2s^D(\gamma)+R_4+R_5
\end{aligned}
\end{equation}
with $s^D(\gamma)$ defined in \eqref{Scottfunction} and
\begin{equation}
  R_4:=Z^2\gamma^{-2}\sum_{l=0}^{L-1}\sum_{j\geq\frac12, j=l\pm\frac12}(2j+1)\sum_{n=K-l+1}^{\infty}(\lambda^S_{\gamma,n,l}-\lambda^D_{\gamma,n,l,j})
\end{equation}
and
\begin{equation}
  R_5:=Z^2\gamma^{-2}\sum_{l=L}^{\infty}\sum_{j\geq\frac12, j=l\pm\frac12}(2j+1)\sum_{n=1}^{\infty}(\lambda^S_{\gamma,n,l}-\lambda^D_{\gamma,n,l,j}).
\end{equation}
By Corollary \ref{everror3}, for every $\gamma\in[0,1)$ there exists a
constant $C>0$ such that
\begin{equation}
  \begin{aligned}
    R_4&\leq C Z^2\gamma^{-2}\sum_{l=0}^{L-1}\sum_{j\geq\frac12,
      j=l\pm\frac12}(2j+1)\sum_{n=K-l+1}^{\infty}\frac{\gamma^4}{\left(n+l\right)^3l}\\
    &\leq C
    \gamma^{2}Z^2\sum_{l=0}^{L-1}\sum_{n=K-l+1}^{\infty}\frac{1}{\left(n+l\right)^3}
    \leq C \gamma^{2}Z^2LK^{-2}=\bO(Z^{17/12}).
  \end{aligned}
\end{equation}
Similarly,
\begin{equation}
\begin{aligned}
 R_5&\leq C Z^2\gamma^{-2}\sum_{l=L}^{\infty}\sum_{j\geq\frac12, j=l\pm\frac12}(2j+1)\sum_{n=1}^{\infty}\frac{\gamma^4}{\left(n+l\right)^3l}\leq C \gamma^{2}Z^2\sum_{l=L}^\infty\sum_{n=1}^{\infty}\frac{1}{\left(n+l\right)^3}\\
 &\leq C \gamma^{2}Z^2\sum_{l=L}^\infty\frac{1}{l^2}
 \leq C \gamma^{2}Z^2L^{-1}=\bO(Z^{23/12}).
 \end{aligned}
\end{equation}

\section{Comparison with (Semi-)Empirical Data and Schwinger's Prediction}

As described in the introduction, quantum electrodynamics is believed
to describe particles interacting through electromagnetic forces, including
heavy atoms. Unfortunately, the numerical
evaluation of such systems seems to be beyond present techniques, not
to mention the principal problem that QED is only perturbatively
defined rendering the treatment of heavy neutral atoms difficult.

In view of this fact, the comparison with experimental values appears
to be the only source of validation of the results. This, however, is
not directly possible, since our asymptotic result requires to fix
$\alpha Z$. Experimentally, though, $\alpha_\mathrm{physical}$ is a fixed
constant of value roughly $1/137$, whereas Z takes integer values up to 120.

Moreover, the published atomic ground state energies
$E_{\mathrm{NIST}}(Z)$ in the NIST Atomic Spectra Data Base
\cite{NIST_ASD} are measured only to a small extent. The majority of
the energies for large $Z$ is extrapolated or computed, i.e.,
assumptions on underlying approximate mathematically uncontrolled
models influence those data. In addition the experimental values
obviously also contain other QED effects not contained in the Furry
Hamiltonian defined through \eqref{energyform}.

Thus it is not obvious that our formula for the ground state energy
with $\gamma$ replaced by $\alpha_{\mathrm{physical}}Z$ should give an
improved quantitative description of large atoms. Nevertheless,
emboldened by Sell's \cite{Stell1977} \textit{principle of
  unreasonable utility of asymptotic expansions} and its history of successful application, e.g., Lebowitz and Waisman \cite{LebowitzWaisman1980}
and Schwinger \cite{Schwinger1980}, we offer a graphical
comparison of
$${E_\mathrm{NIST}(Z)- E_\mathrm{TF}(Z)\over Z^2}$$
with the relativistic Scott function 
$$\frac12-s^D(\alpha_\mathrm{physical}Z)$$
in Figure \ref{Scottpicture}. Additionally we show the Dirac-Hartree-Fock calculations carried out by {Desclaux} \cite{Desclaux1973}. Note that the Scott function is not divergent for 
$Z\to\frac1{\alpha_{\mathrm{physical}}}$ (as was claimed in \cite{Sorensen1998}), but instead approaches the numerical value
\begin{equation}
 \lim_{\gamma\to1}(\frac12-s^D(\gamma))\approx-1.91.
\end{equation}

Unlike the Chandrasekhar model \cite{LiebYau1988} and the
Brown-Ravenhall model \cite{Evansetal1996,Franketal2009} which give
substantially too low energies and break down for $\gamma=2/\pi<1$ and
$\gamma=2/(\pi/2+2/\pi)$ respectively, the Furry picture -- which for
numerical purposes is implemented through the Douglas-Kroll-Hess
transform -- does not only give stable ground state energies and good
numerical values in quantum chemistry (see \cite{ReiherWolf2009} for
an extensive overview), but, as Figure \ref{Scottpicture} shows, it
also offers a step toward a quantitative correct description of heavy
atoms.

\begin{figure}[ht]
    \centering
`    \includegraphics[width=1.0\textwidth]{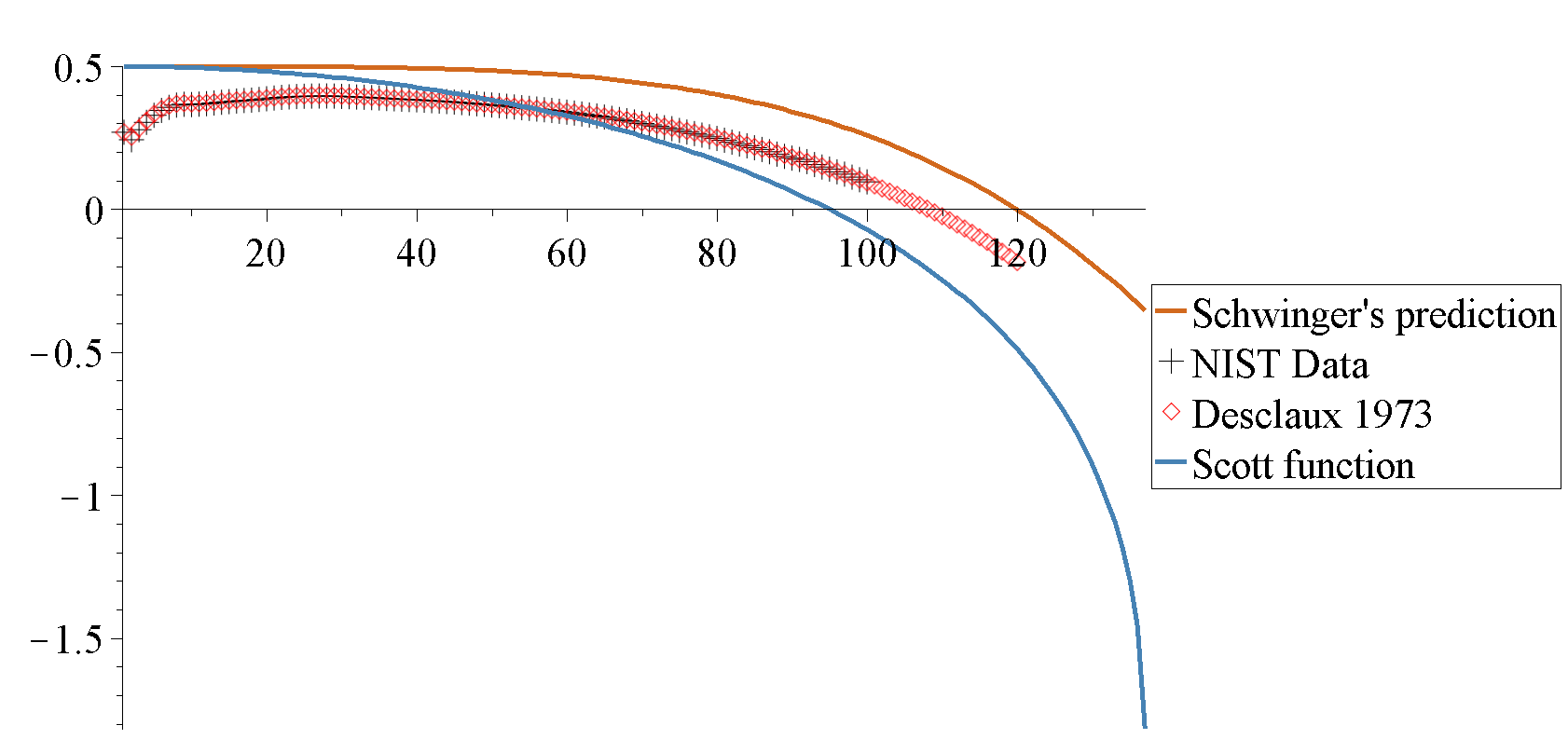}
    \caption{Comparison of the relativistic Scott function with data
      taken from the NIST database \cite{NIST_ASD}, Dirac-Fock
      calculations \cite{Desclaux1973} and Schwinger's original
      prediction \cite{Schwinger1980}.}
    \label{Scottpicture}
\end{figure}

We will supplement the data with Schwinger's prediction
\cite{Schwinger1980} of the relativistic Scott correction which was
derived from the $\gamma^4$ fine structure correction to the
non-relativistic eigenvalues implied by the Dirac equation, i.e., the
second term of the perturbative expansion of the kinetic energy
$$\langle\psi,-\frac{\bp^4}{8}\psi\rangle,$$
the spin-orbit coupling
$$
\langle\psi,\frac{\gamma}2\mathbf{S}\cdot\mathbf{L}\frac1{|\bx|^3}\psi\rangle,$$
and the Darwin term
$$\langle\psi,\frac{\pi}2 \gamma \delta(\bx)\psi\rangle.$$
Together they yield
$$\delta \lambda_{n,l,j}=\frac{\gamma^4}{-2(n+l)^3}\left(\frac1{j+\frac12}--\frac3{4}\frac1{n+l}\right)$$
(cf. \eqref{everror2}). Following Schwinger's computations and using
the identity
$$
\sum_{m=1}^\infty\sum_{n=1}^{\infty}\frac1{(m+n)^s}
=\sum_{m=1}^\infty\sum_{n=1}^m\frac1{m^s}=\zeta(s-1)-\zeta(s),
$$ 
we obtain
\begin{equation}
  \begin{aligned}
    s^D_{Schwinger}:=&\frac1{\gamma^2}\sum_{l=0}^{\infty}\sum_{n=1}^\infty\sum_{j\geq\frac12, j=l\pm\frac12}(2j+1)\delta\lambda_{n,l,j}\\
    =&-\sum_{n=1}^\infty\frac{\gamma^2}{n^3}(1-\frac3{4n})-\sum_{l=1}^\infty\sum_{n=1}^{\infty}\frac{\gamma^2}{(n+l)^3}(2-\frac34)\frac{2l+1}{n+l}\\
    =&-\zeta(3)\gamma^2+\frac3{4}\zeta(4)\gamma^2-\frac5{4}\sum_{l=1}^\infty\sum_{n=1}^{\infty}\frac{\gamma^2}{(n+l)^3}+\frac{3}{4}\sum_{l=1}^\infty\sum_{n=1}^{\infty}\frac{\gamma^2}{(n+l)^4}\\
    =&(-\zeta(3)+\frac3{4}\zeta(4)-\frac5{4}\zeta(2)+\frac5{4}\zeta(3)+\frac{3}{4}\zeta(3)-\frac{3}{4}\zeta(4))\gamma^2\\
    =& (\frac{5\pi^2}{24}-\zeta(3))\gamma^2\sim -0.854\gamma^2
\end{aligned}
\end{equation}
which is asymptotically correct in the non-relativistic limit $\gamma\to0$.

\appendix

\section{Thomas-Fermi Theory}
\label{TFtheory}
We collect some well know auxiliary facts of the Thomas-Fermi
density. For a more exhaustive overview see Lieb and Simon
\cite{LiebSimon1977} and Lieb \cite{Lieb1981}. Consider the minimizer
$\rtf$ of the Thomas-Fermi functional
\begin{equation}
  \mathcal{E}^\mathrm{TF}[\rho]=\frac3{5}\frac{(3\pi^2)^\frac{2}{3}}{2}\int_{\rz^3}\rho(x)^\frac{5}{3}\rd x-\int_{\rz^3}\frac{Z\rho(x)}{|x|}\rd x +\frac12\int_{\rz^3}\int_{\rz^3}\frac{\rho(x)\rho(y)}{|x-y|}\rd x\rd y
\end{equation}
in the set $\mathcal{I}:=\{\rho\in L^\frac{5}{3}|\rho\geq0,\  D[\rho]<\infty\}$.
The infimum
$$E_\mathrm{TF}(Z) := \inf \cE^\mathrm{TF}(\mathcal{I})$$
is attained and $\int \rho =Z$. Moreover
\begin{equation}
  \label{eq:tf}
  E_\mathrm{TF}(Z)= E_\mathrm{TF}(1)Z^{7/3}.
\end{equation}
The Thomas-Fermi energy of hydrogen has been calculated first by Milne
\cite{Milne1927} and later by Sommerfeld \cite{Sommerfeld1932}
reducing it to the slope of the Thomas-Fermi potential at zero, i.e.,
Baker's constant and has the value
\begin{equation}
  \label{eq:values}
  E_\mathrm{TF}(1)\sim -0.768745 [Ha].
\end{equation}

\begin{lemma}[Properties of the Thomas-Fermi density]
  The Thomas-Fermi density $\rtf$ and its mean-field potential
  $V_Z:=\rtf\ast \frac1{|\cdot|}$ satisfy
\begin{enumerate}
\item $\|V_Z\|_{\infty}\leq C Z^{\frac{4}{3}}$ for some constant $C>0$,
\item $|\bp V_Z|\leq \frac{C Z^\frac32}{|\bx|^\frac12}$,
\item $|\bp^2 V_Z|=|\rtf|\leq \frac{C Z^\frac32}{|\bx|^\frac32}$.
\end{enumerate}
\end{lemma}

This is easily verified by using the scaling relation of the
Thomas-Fermi density $\rtf(\bx)=Z^2\rho^{TF}_1(Z^\frac13\bx)$ and the
Thomas-Fermi equation
\begin{equation}
 \frac12(3\pi^2)^\frac23(\rtf)^\frac23=\frac{Z}{|\cdot|}-V_Z.
\end{equation}

\section{Partial Wave Analysis\label{PWA}}

For the convenience of the reader and for normalization of the
notation we gather some fact on the partial wave analysis of the
Brown-Ravenhall operator (see, e.g., \cite{Franketal2009}).

We denote by $Y_{l,m}$ the normalized spherical harmonics on the unit
sphere $\mathbb S^2$ (see, e.g., \cite{Messiah1969}, p. 421) with the
convention that $Y_{l,m}\equiv0$ if $|m|>l$, and we define for $ j\in
\mathbb{N}_0 + \tfrac{1}{2} $, $l \in \mathbb{N}_0 $, and $
m=-j,\ldots,j $ the spherical spinors
\begin{equation} 
  \Omega_{j,l,m}(\omega) :=        \begin{cases}
                \begin{pmatrix} \sqrt{\frac{j+m}{2j}} \,
                Y_{l,m-\eh}(\omega)\\ \sqrt{\frac{j-m}{2j}} \,
                Y_{l,m+\eh}(\omega)
                \end{pmatrix} & \text{if} \ j=l+ \eh,\\[4ex]
                \begin{pmatrix} -\sqrt{\frac{j-m+1}{2j+2}} \, 
                Y_{l,m-\eh}(\omega)\\ \sqrt{\frac{j+m+1}{2j+2}}\, 
                Y_{l,m+\eh}(\omega)
                \end{pmatrix} & \text{if} \ j=l-\eh.
        \end{cases}
\end{equation}
The set of admissible indices is $ \mathcal{I} := \{(j,l,m)\,
| \, j\in\nz-1/2,\ l=j\pm1/2,\; m=-j,...,j \} $. It is known that the
functions $\Omega_{j,l,m}$, $(j,l,m)\in\mathcal I$, form an
orthonormal basis of the Hilbert space $L^2(\mathbb
S^2;\mathbb{C}^2)$. They are joint eigenfunctions of $ \mathbf{J}^2 $, $\mathbf{L}^2$,and $ J_3 $ with eigenvalues given by $ j(j+1) $, $l(l+1)$,
and $m$.  The subspace $\gh_{j,l,m}$ corresponding to the joint
eigenspace of total angular momentum $ \mathbf{J}^2 $ with eigenvalue
$ j(j+1) $ and angular momentum $\mathbf{L}^2 $ with eigenvalue $
l(l+1) $ is given by
$$\gh_{j,l,m} = {\rm span}\{ \bx \mapsto |\bx|^{-1} \, f(|\bx|) \,
\Omega_{j,l,m}(\omega_\bx) \ | \ f \in L^2(\mathbb{R}_+)\}
$$ 
where $\omega_\bx:=\bx/|\bx|$. This leads to the orthogonal decomposition
\begin{equation}\label{eq:decompapp}
  \gh = \bigoplus_{j \in \nz_0 + \frac{1}{2} } \bigoplus_{l=j \pm 1/2} \gh_{j,l},
\qquad
\gh_{j,l} = \bigoplus_{m=-j}^j \gh_{j,l,m},
\end{equation}
of the Hilbert space of two spinors.

Note the identity (see, e.g., Greiner \cite{Greiner1990})
\begin{equation}
(\bsigma \cdot \bomega_\bx) \Omega_{j,l,m}(\bomega_\bx)=-\Omega_{j,2j-l,m}(\bomega_\bx). 
\end{equation}
Furthermore for fixed $l$
\begin{equation}
 \text{dim}_l(\{ \Omega_{j,l,m}, (j,l,m)\in \mathfrak{I}\})=2(2l+1).
\label{diml}
\end{equation}
In $\mathfrak{H}=L^2(\rz^3,\cz^4)$ the corresponding decomposition is
obtained from the discussion above by noting that in
$L^2(\mathbb {S}^2;\mathbb{C}^4)$ eigenfunctions corresponding to the
operators $\bJ^2, J_z$ (now acting on 4-spinors) can be constructed as
linear combinations of the 2-spinor eigenfunctions
\begin{gather}
\begin{aligned}
 &\Phi^+_{j,l,m}=\begin{pmatrix} \ri \Omega_{j,l,m} \\ 0 \end{pmatrix} \\
 &\Phi^-_{j,l,m}=\begin{pmatrix}  0 \\ -\Omega_{j,2j-l,m} \end{pmatrix},
\end{aligned}
\end{gather}
i.e., 
\begin{equation}
 \mathfrak{H}=\bigoplus_{j,l}\gH_{j,l}:=\bigoplus_{(j,l,m)\in \mathfrak{I}}\gH_{j,l,m}
\end{equation}
where
\begin{equation}
  \mathfrak{H}_{j,l,m}=\{\bx\to\frac1{|\bx|}f(|\bx|)\Phi^+_{j,l,m}(\omega_x)+\frac1{|\bx|}g(|\bx|)\Phi^-_{j,l,m}(\omega_x); f,g\in L^2(\rz_+)\}.
\end{equation}
The orthogonal projections onto those spaces are denoted by $\Pi_{j,l}$
and $\Pi_{j,l,m}$.

While the  Dirac operator  is not commuting  with the  orbital angular
momentum  $\bL$,  these subspaces  of  $\mathfrak{H}$  are still  left
invariant since
\begin{gather}
\begin{aligned}
 &\ri(\balpha\cdot\bomega_\bx)  \Phi^+_{j,l,m}(\bomega_\bx)=-\Phi^-_{j,l,m}(\bomega_\bx)\\
 &\ri(\balpha\cdot\bomega_\bx)  \Phi^-_{j,l,m}(\bomega_\bx)=+\Phi^+_{j,l,m}(\bomega_\bx)
\end{aligned}
\end{gather}
(Balinsky and Evans \cite[p. 32, 2.1.30]{BalinskyEvans2011}).

\section{The Foldy-Wouthuysen Transformation}
\label{FWT}
The free Dirac operator my defined using two unitary transforms, the
Fourier transform and the Foldy-Wouthuysen transform
\cite{FoldyWouthuysen1950}
\begin{equation}
u(\bp)=a_+(p)\mathbf{1}+a_-(p)\beta \balpha\cdot\bomega
\end{equation}
using $p=|\bp|$, $\bomega=\frac{\bp}{p}$, $a_\pm=\sqrt{\frac{E(p)\pm
    1}{2E(p)}}$, $E(p)=\sqrt{p^2+1}$, through the following formula
\begin{equation}
\begin{aligned}
D_0:=&\cF^{-1}\begin{pmatrix}1 & c\bsigma\cdot{\bp}\\ c\bsigma\cdot \bp & -1\end{pmatrix}\cF\\
=&\cF^{-1}u(\bp/c)^{-1}
  \begin{pmatrix} E(p/c) & 0 & 0 & 0\\ 0 & E(p/c) & 0 & 0 \\ 0 & 0 & -E(p/c) & 0\\ 0 & 0 & 0 & -E(p/c) 
  \end{pmatrix} 
  u(\bp/c)\cF.
\end{aligned}
\end{equation}
Since the diagonal operator is self-adjoint with domain
$L^2(\rz:\cz^4,(1+\bp^2)\rd\bp)$, $D_0$, the free Dirac operator is
self-adjoint on $H^1(\rz^3,\rd \bp)$.

For the Brown-Ravenhall operator this gives rise to the isometry
\begin{equation}
\begin{aligned}
\mathbf{\Phi}_c:\  \mathfrak{h}:=&L^2(\rz^3:\cz^2)&\mapsto &\mathfrak{h}^B\\
&\psi &\mapsto &(\Phi_0\psi,\Phi_1\psi)
\end{aligned}
\end{equation}
where in Fourier representation $\Phi_0(\bp)=\sqrt{\frac{E(p)+
    1}{2E(p)}}$, $\Phi_1(\bp)=\sqrt{\frac{E(p)-
    1}{2E(p)}}\bsigma\cdot\bomega$. In particular, the isometric
property is easy to see since the map
\begin{equation}
 \ba \mapsto (\Phi_0(\bp)\ba,\Phi_1(\bp)\ba)
\end{equation}
for $\ba\in\cz^2$ is an isometry from $\cz^2$ to $\cz^4$ independent of $\bp$.

\section{The Trial Density Matrix}
\label{TDM}
In this section we will define the density matrix $d$ which will be
used to bound the Furry energy from above via the usual min-max
principle. The construction will be in the spirit of
\cite{SiedentopWeikard1986,SiedentopWeikard1987O}, in particular the
density matrix will be split into two parts, corresponding to low and
high angular momenta respectively, i.e.,
\begin{equation}
 d:=d_<+d_{>,\phi}=d_<+\mathbf{\Phi}_c d_> \mathbf{\Phi}_c^*
\end{equation}
acting on $L^2(\rz^3:\cz^4)$ and its non-relativistic version
\begin{equation}
 d^S:=d^S_<+d_>
\end{equation}
which is acting on $L^2(\rz^3:\cz^2)$.  Low angular momenta correspond
to orbits with perihelions close to the nucleus, while high angular
momenta prohibit orbits in the near vicinity of the nucleus. The
separation will be set at $L:=[Z^{\frac1{12}}]$.

Note that for the construction of the relativistic trial density
matrix we lift the high angular momentum part of the non-relativistic
density matrix from 2-spinor to 4-spinor space via the Foldy
Wouthuysen transformation (cf. appendix \ref{FWT}). Consequently $d$
neither belongs to the positive spectral subspaces of the hydrogenic
nor the free Dirac operator. However, by construction, the individual
parts $d_<$ and $d_{>,\phi}$ do.
\subsection{Low Angular Momenta}

Close to the nucleus the electron-electron interaction will be dominated by the nuclear interaction. This is reflected by choosing the unmodified eigenfunctions of the Coulomb-Dirac operators for small angular momenta,
\begin{equation}
 d_<:=\sum_{l=0}^{L-1}d_{<,l},\ \ \ d_{<,l}=\sum_{j\geq\frac12, j=l\pm\frac12} d_{<,j,l}
\end{equation}
where
\begin{equation}
d_{<,j,l}=\sum_{m=-j}^j\sum_{n=1}^{K-l}\left|\psi_{j,l,m,n}\rangle\langle\psi_{j,l,m,n}\right|
\end{equation}
with $K=[\const Z^\frac13]$ and $\psi_{j,l,m,n}$ being eigenfunctions
of the Coulomb-Dirac operator $
c\balpha\cdot\bp+c^2\beta-c^2-\frac{Z}{|\bx|}$ corresponding to
eigenvalue $c^2\lambda^D_{Z,n,l,j}$ (cf. \eqref{l2}) and azimuthal
quantum number $m$.  The cutoff $K$ was set on the order of the last
occupied shell of the Bohr atom. The Dirac eigenfunctions can be
computed explicitly and can be found for instance in
\cite{Thaller1992}. Likewise, we define the non-relativistic analogue
\begin{equation}
 d^S_<:=\sum_{l=0}^{L-1}d_{<,l},\ \ \ d_{<,l}=\sum_{j\geq\frac12, j=l\pm\frac12} d^S_{<,j,l}
\end{equation}
with
\begin{equation}
d^S_{<,j,l}=\sum_{m=-j}^j\sum_{n=1}^{K-l}\left|\psi^S_{j,l,m,n}\rangle\langle\psi^S_{j,l,m,n}\right|
\end{equation}
where $\psi^S_{j,l,m,n}$ denote the eigenfunctions of $\frac{\bp^2}{2}
-\frac{Z}{|\bx|}$ on $L^2(\rz^3,\cz^2)$ with eigenvalues
$\lambda^S_{Z,n,l}$ (see \eqref{l2}).

\subsection{High Angular Momenta}
For large angular momenta, the electrons are moving slowly at far
distance from the nucleus. Moreover, the correspondence principle
demands a quasi-classical behavior for large quantum numbers. These
heuristics suggest that for large angular momenta a semi-classical and
non-relativistic description of the electrons in their mean field (if
assumed interacting) and a description by the unscreened electrons for
small angular momenta suffices to obtain not only the leading
contribution to the ground state energy (Thomas-Fermi) but also its
first correction (Scott), if the cut between the large and high angular
momenta is chosen appropriately. This idea originates in
\cite{SiedentopWeikard1986} and was guiding the construction of trial
density matrices in
\cite{SiedentopWeikard1987O,Franketal2007S,Franketal2008}. We will
follow it also here. Of course, the same heuristics additionally
suggests that the difference between the Furry and Brown-Ravenhall
projections will be small for large angular momenta. That this is
indeed the case is the essential technical contribution of this paper which
allows to treat the Furry picture. We thus choose
\begin{equation}
  d_>:=\sum_{l\geq L}d_{l},\ \ \ d_l:=\sum_{j=l\pm\frac12}\sum_{m=-j}^{m=j}\sum_{n\in\gz}w_{n,l}\left|\phi_{n,l}\Omega_{j,l,m}\rangle\langle\phi_{n,l}\Omega_{j,l,m}\right|.
\end{equation}
The $\phi_{n,l}$ are the Macke orbitals as constructed in \cite{SiedentopWeikard1987O}.

The kinetic energy estimate for the density matrix $d_>$ found in
\cite{Franketal2009} will be useful:
\begin{lemma}[\cite{Franketal2009}, Lemma E.1]
\label{lemmae1}
Let $L=\left[Z^{1/12}\right]$. Then for large $Z$,
\begin{equation}
\sum_{l=L}^\infty l^{-2} \tr(\bp^2 \mathbf{\Phi}_c d_> \mathbf{\Phi}_c^*)=\sum_{l=L}^\infty l^{-2} \tr(\bp^2 d_l)=\bO(Z^2/L).
\end{equation}
\end{lemma}

\textsc{Acknowledgment:} It is a pleasure to acknowledge partial
support of the Institut Henri Poincar\'e through its program
``Variational and Spectral Methods in Quantum Mechanics'' and the
Deutsche Forschungsgemeinschaft through its TR-SFB 12 (Symmetrien und
Universalit\"at in Mesoskopischen Systemen).

%\bibliographystyle{plain}
%\bibliography{coulomb}

\def\cprime{$'$}

\end{document}